\documentclass[journal]{IEEEtran}
\usepackage{amsmath, amssymb}
\usepackage{graphicx, color}
\usepackage{bm}
\usepackage{cite}
\graphicspath{{figure/}}

\newtheorem{defi}{\bf Definition}
\newtheorem{rmk}{\bf Remark}

\newtheorem{thm}{\bf Theorem}

\begin{document}

\title{A Multiobjective Approach to Multi-microgrid System Design}

\author{Wei-Yu~Chiu,~\IEEEmembership{Member,~IEEE}, Hongjian Sun,~\IEEEmembership{Member,~IEEE},  and
        H.~Vincent~Poor,~\IEEEmembership{Fellow,~IEEE} 
\thanks{This work was supported in part by the Ministry of Science and Technology of Taiwan under Grant 102-2218-E-155-004-MY3,
and in part by the U.S. National Science Foundation under Grant CMMI-1435778.
The research leading to these results has received funding from the European Commission's Horizon 2020 Framework Programme (H2020/2014-2020) under grant agreement No. 646470, SmarterEMC2 Project.}
\thanks{W.-Y. Chiu is with the Department of Electrical Engineering and with
the Innovation Center for Big Data and Digital Convergence, Yuan Ze University, Taoyuan 32003, Taiwan
        (e-mail: chiuweiyu@gmail.com).}%
\thanks{H. Sun is with the School of Engineering and Computing Science, Durham University, DH1 3LE, U.K. (e-mail: mrhjsun@hotmail.com).}
\thanks{H.~V.~Poor is with the Department of Electrical Engineering, Princeton University, Princeton, NJ 08544, USA (e-mail: poor@princeton.edu).}%
\thanks{
\copyright 2015 IEEE. Personal use of this material is permitted. Permission from IEEE must be
obtained for all other uses, in any current or future media, including
reprinting/republishing this material for advertising or promotional purposes, creating new
collective works, for resale or redistribution to servers or lists, or reuse of any copyrighted
component of this work in other works.
}
\thanks{Digital Object Identifier 10.1109/TSG.2015.2399497}
}
\maketitle

\begin{abstract}

The main goal of this study is to design  a market operator (MO) and a distribution network operator (DNO) for a network of microgrids
in consideration of  multiple objectives.
This is a high-level design and
only those microgrids with nondispatchable renewable energy sources are considered.
For a power grid in the network, the net value derived from providing power to the network must be maximized.
For a microgrid,
it is desirable to maximize the net gain derived from consuming the received power.
Finally, for an independent system operator,  stored energy levels at microgrids
must be maintained as close as possible to storage capacity to secure network emergency operation.
To achieve these objectives, a multiobjective approach is proposed:
the price signal generated by the MO and  power distributed by the DNO are assigned based on
a Pareto optimal solution of a multiobjective optimization problem.
By using the proposed approach, a fair scheme that does not advantage one particular objective can be attained.
Simulations are provided to validate the proposed methodology.
\end{abstract}


\section{Introduction}

Dynamic pricing in the smart grid is an approach that helps
reshape or reduce the power demands by varying the cost of power service over time.
Power consumers who are sensitive
to the energy price may change their power use in response to the varying price signals~\cite{bu2011dynamic}.
Dynamic pricing has been extensively discussed and explored in the literature~\cite{Alva3,Alva5,chiu_sg,kok2011dynamic,huhardware,SGComm_Nov_12,shift1,shift2,peak_load,mohsenian2010optimal,wu2011demand,mohsenian2010autonomous,vytelingum2010agent,voice2011decentralised}.
For instance, a pricing scheme can be used as an area control method~\cite{Alva3,Alva5,chiu_sg,kok2011dynamic}
or a home-scale method~\cite{huhardware}  for energy management.
 A grid node may store energy in a local energy storage system when the price is low,
and use the stored energy when the price increases~\cite{SGComm_Nov_12}.
When shiftable loads are involved in the grid,  power users may vary their power demands according to the price~\cite{shift1,shift2}.
This further indicates that the pricing scheme is capable of lowering the peak load~\cite{peak_load} so that the maximum system capacity and thus the cost
 can be reduced.

For a successful application of dynamic pricing,
 elements such as
 active participation of consumers in demand response~\cite{mul_agent1}, robust energy management~\cite{EMS1,EMS0}, and
proper power distribution~\cite{EMS2}
 are  needed in the smart grid.
To include these elements for further investigation on the grid, we  consider
a network of microgrids, an independent system operator (ISO) that consists of
a market operator (MO) and a distribution network operator (DNO), and a power grid~\cite{mul_agent}.
In the network, microgrids are locally connected with nondispatchable renewable energy sources (RESs)~\cite{int_RES2},
and some of them  possess energy storage systems.

The microgrids take advantage of  the power supply from the power grid, RESs,
and their local energy storage systems to meet their respective power demand.
The MO generates a price signal that is related to the power supply, power demand and stored energy in the network.
The DNO distributes the power generated at the power grid to microgrids.
The MO and DNO in the ISO are positioned at the highest control level in the multi-microgrid environment,
which allows these two entities to have a substantial impact on the grid operation.

To facilitate ensuing discussions, we refer to the network and the participants of the network as the multi-microgrid system.
Furthermore, we refer to
determining the price signal generated by the MO
and the power distribution performed by the DNO as  a high-level design of the multi-microgrid system, i.e.,
the multi-microgrid system design includes both an electricity market design and a power distribution design.
In the multi-microgrid system, the power grid and microgrids desire their respective asset utilization to be optimized,
and the ISO desires that the multi-microgrid system should store as much energy as possible to secure emergency operation.
From this perspective,  an optimal multi-microgrid system design leads to  optimizing interests of the microgrids, power grid, and ISO simultaneously.
This suggests the consideration of a multiobjective optimization problem (MOP).

In contrast,
 most existing  approaches related to  high-level grid system designs
  consider solely utility maximization
 by maximizing certain aggregate benefits, sometimes termed the social welfare~\cite{price_em}.
For instance, the associated aggregate function  can be a sum of all utility functions of microgrids minus the cost of power generation at the power grid~\cite{samadi2010optimal,tarasak2011optimal}.
 These formulations lead to a single-objective optimization problem (SOP).
Market price generation,  power distribution, or efficient energy consumption can then be characterized as
the solution to the SOP.

Although able to produce a reasonable design,
 existing approaches can suffer from at least one of the following drawbacks.
First, optimizing a system performance index related to emergency operation is often neglected.
Maximizing the aggregate utility of the power grid and microgrids is the only objective that must be attained.
Second, the mutual relationship between   the utility maximization of the power grid and of the microgrids
  is seldom addressed. Third, system operation schemes resulting from maximizing the aggregate function may favor
a particular participant~\cite{marler2004survey,messac2000a,messac2000b,das1997closer}\footnote{In general, to avoid favoring a particular participant, the weighting coefficients used in the aggregate function can be set equal to each other and normalized objective functions can be used. However, such a method is invalid when the MOP in which the objectives  form the aggregate function of participants has a nonconvex Pareto front (PF).
 Even if the PF is convex, other difficulties can be introduced by using aggregate functions derived from different weights to approximate the PF, as discussed in the references.}, e.g., the power grid or microgrids, depending on the weighting coefficients used in the aggregate function and on the conflicting relationship among objectives of participants.
In a fair and transparent setting  joined by active participants, this bias should be generally avoided.

To address these drawbacks, we propose a multiobjective approach to multi-microgrid system design.
Three objectives are considered, introducing a three-dimensional objective function space\footnote{We use the terms decision variable space and objective function space to describe
the domain and codomain, respectively,  of the vector-valued objective function consisting of the three  scalar-valued  objective functions.}.
The first objective is to maximize the overall net value derived from consuming power at microgrids, i.e.,  the utility maximization for the microgrids.
The second objective is to maximize net revenue derived from providing power at the power grid, i.e., the utility maximization for the power grid.
Finally, to secure the emergency operation,
the third objective is to maximize
a sum of the  stored energy levels within the multi-microgrid network,
 corresponding to maximizing the interest of the ISO.
The consideration of the third objective in our formulation  addresses  the first drawback  mentioned previously.

The proposed multiobjective approach leads to solving an MOP.
 Since multiple objectives are involved, Pareto domination is adopted.
A multiobjective immune algorithm (MOIA) is developed to solve the problem by
searching for feasible points in the decision variable space that represent prices and the amount of power distributed to microgrids.
 During the solving process, dominated or infeasible points are gradually removed.
 In other words, nondominated and feasible points are maintained, yielding a set of approximated Pareto optimal solutions at the end of the process.
Each solution associates with an approximated Pareto optimal design.
 The whole set corresponds to an approximated Pareto front (APF).
 The APF is of importance because it can clearly illustrate how one objective affects the others, which
  cannot be achieved when an SOP is formulated and solved for utility maximization.
The ability to produce the APF addresses the second drawback.

Based on the APF, a design for the DNO and MO that does not favor a particular participant can be derived.
In our framework, if a vector on the APF has an entry with an extreme value, then the associated design favors one particular participant.
To achieve a fair design, we fully explore the obtained APF:
An optimization process is performed over the improvement of the associated objective function values in three dimensions.
The vector on the APF that maximizes the minimum value of the normalized improvement in all dimensions is selected.
The associated solution is then used to characterize the design.
The exploration of the APF that leads to a fair design addresses the third drawback.

 The main contributions of this paper are summarized as follows.
 The proposed methodology for multi-microgrid system designs can avoid the drawbacks introduced by using aggregate functions.
 To the best of our knowledge, in recent literatures related to the multi-microgrid system design, few efforts have been devoted to addressing a multiobjective
approach for price generation and power distribution. We thus propose a multiobjective formulation, which can provide a framework for future exploration of multiobjective methodology in related fields.
As the formulation leads to solving  an MOP,
we develop an MOIA  so that an APF can be produced.
Finally, we devise a simple method that can yield
 a fair multi-microgrid system design based on the APF obtained.

The rest of this paper is organized as follows. Section~\ref{sec_SystModel} describes
the  network components. In Section~\ref{sec_ProbFormu}, a multiobjective formulation is proposed, leading to an MOP, and related analysis is performed.
In Section~\ref{sec_design}, an algorithm that can solve the MOP is developed. Simulation results are presented in  Section~\ref{sec_sim}. Finally, Section~\ref{sec_con} concludes this paper.

\begin{figure}
  \centering
  \includegraphics[width=8.5cm]{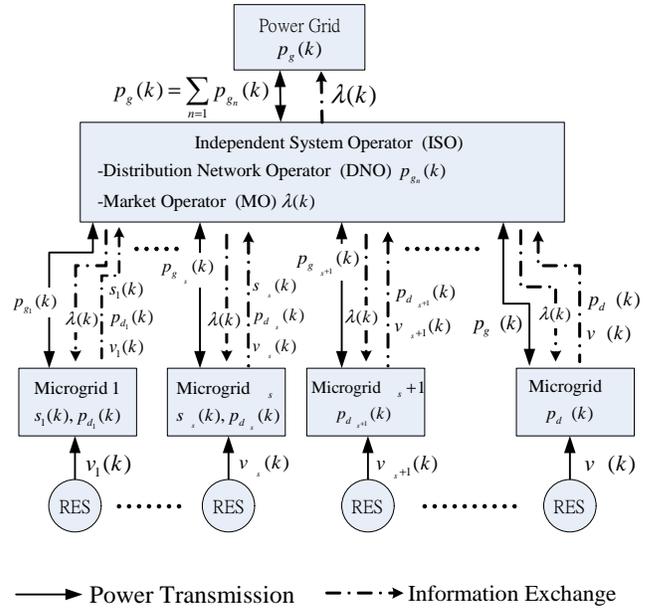}\\
  \caption{System operation model of a network of microgrids. }\label{fig_Control_level}
\end{figure}

\section{System Model}\label{sec_SystModel}

\begin{table}
  \centering
  \caption{Symbols}\label{tab_notation}
  \hspace{-0.2cm}
\begin{tabular}{cl}
  \hline
\hline
  Symbol & Description \\
\hline
  $N$ & the number of microgrids \\
  $\mathcal{N}$ &  the set of microgrid indices, i.e., $\mathcal{N}=\{1,2,...,N\}$ \\
  $N_s$ & the number of microgrids that have a local energy \\
  &    storage system\\
  $\mathcal{N}_s$ & the index set for the microgrids that have a local \\
  &   energy storage system (we assume $\mathcal{N}_s=\{1,2,...,N_s\}$)\\
  $k$ & time index \\
$s_n(k)$  & stored energy level of $n$th microgrid at time $k$\\
$\bar{s}_n$ and $\underline{s}_n$ & maximum storage capacity and secure energy level  for\\
&   emergency operation, respectively  ($\bar{s}_n\geq     s_n(k) \geq \underline{s}_n>0$) \\
$\Delta s_n$  & maximum rate of storage charging and discharging \\
&($| s_n(k+1)- s_n(k)| \leq \Delta s_n$) \\
$v_n(k)$  & power supply from nondispatchable RESs to $n$th microgrid \\
&    at time $k$  \\
 $p_{g_n}(k)$  &  power distribution between the power grid and the $n$th\\
& microgrid at time $k$  \\
$p_g(k)$  &   total power distribution, i.e.,  $p_g(k)=\sum_{n=1}^N p_{g_n}(k)$\\
 $p_{d_n}(k)$  & power demand of the $n$th microgrid at time $k$ \\
$\lambda (k)$ & market price at time $k$ \\
$U_g(\cdot)$ &  utility function of the power grid \\
$U_d(\cdot)$ &  utility function of microgrids \\
$U_c(\cdot)$ &  constraint function \\
$[\cdot]_i$ & the $i$th entry of a vector  \\
$:=$ & assignment operator  \\
  \hline
\hline
\end{tabular}
\end{table}

This section describes the system operation model of a network of $N$ microgrids in Fig.~\ref{fig_Control_level}.
Without loss of generality, we assume that there is no power exchange between microgrids because
if some microgrids are connected with links and can exchange power, then we simply combine them as one microgrid.
An information and communication technology (ICT) system has been implemented so that network information, e.g., the price signal and power demand, can be exchanged among the microgrids,  power grid, and  ISO, consisting of the MO and DNO.
In this study, we consider a high-level multi-microgrid system design.
Because the balance of current and voltage is managed at a lower level,
the associated balance equations are not explicitly shown in the following discussions.

For clarity,
the section is divided into three subsections, and Table~\ref{tab_notation}  summarizes the notation used throughout this paper.
Sections~\ref{sub_micro}, \ref{sub_power_grid}, and \ref{sub_ISO} describe mathematical models of the microgrids, power grid, and ISO, respectively.
The respective utility maximization problem is examined.

\subsection{Microgrids}\label{sub_micro}

Let $\mathcal{N}=\{1,2,...,N\}$ be the set of microgrid indices and $\mathcal{N}_s$ be the index set
for the microgrids that have a local energy storage system.
Without loss of generality, we assume $\mathcal{N}_s=\{1,2,...,N_s\}$ where $N_s\leq N$.
Referring to Fig.~\ref{fig_Control_level}, we let
 $p_{g_n}(k)$ denote the power transmission between the power grid and  microgrid $n$.
 If  $p_{g_n}(k)>0$, then microgrid $n$ receives power from the power grid; otherwise,
 microgrid $n$ sells power to the grid.

If $n \in \mathcal{N} \setminus \mathcal{N}_s$, i.e., microgrid $n$ does not possess an energy storage system, then
we have
 \begin{equation}\label{eq_no_storage}
  p_{g_n}(k)- p_{d_n}(k) +v_n(k)=0
\end{equation}
where $p_{d_n}(k)>0$ represents the power demand, and~$v_n(k)>0$
represents the power generated from the RESs, such as solar panels or wind turbines.
In our model, we consider microgrids that have only nondispatchable RES penetration (without dispatchable DG units).

If $n \in \mathcal{N}_s$, then
$s_n(k)$  denotes  the stored energy at microgrid $n$ and  satisfies
\begin{equation}\label{eq_storage_capacity}
0 \leq  s_n(k)\leq \bar{s}_n
\end{equation}
where~$\bar{s}_n$ represents the maximum storage capacity.
The maximum capacity is determined by the media used to store energy.
For instance, if batteries are used to construct the energy storage system,
then the maximum capacity may depend on the chemicals in the batteries and the size of the batteries.
The dynamics of the energy storage system can be expressed as
 \begin{equation}\label{eq_s}
    s_n(k+1)= s_n(k)+  p_{g_n}(k)- p_{d_n}(k) +v_n(k).
\end{equation}
The associated limits on the rate of charging and discharging can be modeled as
 \begin{equation}\label{eq_charge_rate}
   | s_n(k+1)- s_n(k)| \leq \Delta s_n.
 \end{equation}

In grid operation, we consider shiftable loads~\cite{shift1,DSM1} and model the power demand $p_{d_n}(k)$  in~(\ref{eq_no_storage}) and~(\ref{eq_s}) as
 \begin{equation}\label{eq_p_d}
      p_{d_n}(k)= f_{d_n}(\lambda(k),b_n(k))
\end{equation}
where $b_n(k)>0$ represents a nominal value of the base load.
Since in practice base loads have little elasticity, high-accuracy  load forecasting can be attained and we thus assume
$b_n(k)$ in~(\ref{eq_p_d}) is a known quantity.

To maximize the overall utility of microgrids, we consider
\begin{equation}\label{eq_Obj_d}
    \max_{\lambda(k)}  \;U_{d}(p_{d_1}(k), \ldots ,p_{d_N}(k),\lambda(k))
\end{equation}
where the utility function~$U_{d}(\cdot)$ in~(\ref{eq_Obj_d}) represents the net value derived from consuming power~$p_{d_n}(k),n=1,2,\ldots,N,$ when the price is~$\lambda(k)$.

\subsection{Power Grid}\label{sub_power_grid}

Let~$p_{g_n}(k)$ be the power distributed between the power grid and  microgrid $n$, and denote
\begin{equation}\label{eq_dis}
    p_g(k)=\sum_{n=1}^N p_{g_n}(k).
\end{equation}
Let~$U_{g}(p_{g}(k),\lambda(k))$ denote  the utility function of the power grid.
To maximize the interest of the power grid, we consider
\begin{equation}\label{eq_Obj_g}
    \max_{p_{g_n}(k),\lambda(k)} \;U_{g}(p_{g}(k),\lambda(k)).
\end{equation}
The utility function~$U_{g}(\cdot)$ in~(\ref{eq_Obj_g})
is interpreted as the net value derived from~$p_{g}(k)$ when the price is~$\lambda(k)$.

\subsection{ISO}\label{sub_ISO}

Eastern U.S. ISOs such as New York ISO and PJM Interconnection provide emergency demand  response programs~\cite{cappers2010demand}.
To feature those programs in our model,
we determine $\lambda(k)$ and $p_{g_n}(k)$
by solving
\begin{equation}\label{eq_Obj_iso}
\begin{split}
  \max_{\lambda(k),p_{g_n}(k)}  &    \; \sum_{n =1   }^{N_s} s_n(k+1) \\
  \mbox{subject to }  & \; \underline{s}_n \leq s_n(k+1) < \bar{s}_n,n=1,2,...,N_s
\end{split}
\end{equation}
where $\underline{s}_n>0$ represents the minimum energy level required for emergency operation.
In~(\ref{eq_Obj_iso}), maximizing the sum of stored energy levels
 aims to  store as much energy within the multi-microgrid network as possible so that  a safer secure level can be achieved.

\begin{rmk}\label{rmk_sn_func}
According to~(\ref{eq_s}),
 the value of $s_n(k+1)$ depends on the values of~$p_{g_n}(k)$ and  $p_{d_n}(k)$, and
the value of $p_{d_n}(k)$ further depends on the value of $\lambda(k)$, as shown in~(\ref{eq_p_d}).
Therefore, the distributed power $p_{g_n}(k), n=1,2,...,N_s,$ and the price $\lambda(k)$
are the decision variables in~(\ref{eq_Obj_iso}).
\end{rmk}

\begin{rmk}
A scenario related to the model described by Fig.~\ref{fig_Control_level} is discussed as follows.
In a few universities of Taiwan, students who live in dormitories are required to pay their own electricity bills for the use of washing machines, dryers, air conditioner, light, etc. At present, a fixed price for electricity is adopted.
If the price can change over time
and is transparent to students, then
students may adjust their behavior so that
less/more power is consumed when the price is high/low.
In this case, shiftable load, e.g., running washing machines and dryers, can be rearranged.
A win-win situation can thus be created: universities can benefit
 by leveling the load curve, which reduces costs and energy losses within the campus,
  while the students can reduce their expenses.

In the scenario,
 universities act as an agent who buys electricity from a power company and then sells it to students with a time-varying price.
A dormitory implemented with batteries
can be regarded as a microgrid, which possesses certain storage capacity.
Some universities of Taiwan already have solar panels installed on rooftops, which connects
microgrids with RESs.
Most students in Taiwan have smart phones, providing a foundation for building up
the underlying ICT system.
A real-time price for electricity can thus be readily accessed if desired.
Once an advanced metering infrastructure (AMI) system has been implemented in universities,
with the help of the aforementioned facilities, the office of general affairs can function as the ISO.
(As a matter of fact, Yuan Ze University in Taiwan already possesses an AMI system~\cite{EMS}. In the UK, smart meters are being massively rolled out. By 2020, British homes can start to benefit  from smart pricing~\cite{OFGEM,UK_policy}.
In the US,  experimental projects using online feedback systems to help reduction in demand in dormitories have been launched~\cite{petersen2007dormitory,brewer2011kukui,odom2008social}.)
Therefore, the model is practical and suits the smart grid environment that involves active participants.

\end{rmk}

\section{Problem Formulation and Related Analysis}\label{sec_ProbFormu}

This section presents a multiobjective formulation for the multi-microgrid system design, leading to an MOP.
Related analysis on the MOP is performed. In our scenario, a real-time market settlement is considered, and the associated optimization is performed hourly.
 In this setting, the RES output $v_n(k)$ can be forecasted with high accuracy and thus is assumed to be known during the optimization process.

To address the objectives (\ref{eq_Obj_d}), (\ref{eq_Obj_g}), and~(\ref{eq_Obj_iso})   simultaneously,
we formulate the multi-microgrid system design problem
 as the MOP
\begin{equation}\label{eq_MO}
 \begin{split}
  &   \min_{\lambda(k)} \;   -U_{d}(p_{d_1}(k), \ldots ,p_{d_N}(k),\lambda(k)) \\
  & \min_{\lambda(k),p_{g_n}(k)} \;    -U_{g}(p_{g}(k),\lambda(k))\\
  &  \min_{\lambda(k),p_{g_n}(k)} \;    - \sum_{n =1}^{N_s} s_n(k+1)   \\
 \end{split}
\end{equation}
subject to
\begin{equation}\label{eq_overall_s_cst}
     | s_n(k+1)- s_n(k)| \leq \Delta s_n, \underline{s}_n \leq  s_n(k+1)\leq \bar{s}_n, n \in \mathcal{N}_s \\
\end{equation}
 where
\begin{equation*}
\begin{split}
&   s_n(k+1)= s_n(k)+  p_{g_n}(k)- p_{d_n}(k) +v_n(k) , n\in \mathcal{N}_s \\
&   \underline{s}_n  \leq  s_n(0)\leq \bar{s}_n , n\in \mathcal{N}_s \\
&   p_{g_n}(k)- p_{d_n}(k) +v_n(k)=0  , n\in \mathcal{N}\setminus \mathcal{N}_s  \\
&   p_{d_n}(k)= f_{d_n}(\lambda(k),b_n(k)) \mbox{ and }\\
&   p_g(k)=\sum_{n=1}^N  p_{g_n}(k) .
\end{split}
\end{equation*}

In~(\ref{eq_MO}), the decision variables~$\lambda(k)$ and $p_{g_n}(k),n=1,2,\ldots,N_s,$ related to the MO and DNO, respectively, must be determined.
When the dynamic pricing, RESs, and energy storage system are implemented during the same period,
it is reasonable to consider the implementation costs of the RESs and the storage system when the price is adjusted dynamically.
However, if they were constructed in different projects launched in distinct time periods, then
 those costs and the dynamic pricing may become  less relevant.
In our formulation, we consider the situation in which
infrastructures such as the RESs and energy storage system already exist.
 We regard the dynamic pricing
as a high-level method for energy management using the existing infrastructures.
The costs of implementing RESs and storage systems are thus not included in the objectives of~(\ref{eq_MO}).

We interpret the  system stability  as the feasibility of the MOP described by~(\ref{eq_MO}) and~(\ref{eq_overall_s_cst}) for all~$k$.
If the MOP is feasible, then it is possible to find a price and a way  to distribute power generated from the power grid to the microgrids in each time slot
so that the underlying physical constraints are satisfied.
Under this interpretation, we have the following result.

\begin{thm}\label{thm_stable}
 The multi-microgrid system is stable.
\end{thm}
\begin{proof}
Consider the power distribution~$p_{g_n}(k)=p_{d_n}(k) -v_n(k)$ for all $k$ and all $n \in \mathcal{N}$.
We thus have
$  s_n(k+1)=s_n(k), n \in \mathcal{N}_s$.
 Since $\underline{s}_n \leq  s_n(0)\leq \bar{s}_n$, the conditions in~(\ref{eq_overall_s_cst}) are satisfied.
Therefore,  regardless of the value of~$\lambda(k)$ being assigned, $\lambda(k)$ and $p_{g_n}(k)=p_{d_n}(k) -v_n(k),n=1,2,...,N_s,$ form a feasible point of the MOP, i.e.,
  the multi-microgrid system is stable.
\end{proof}

To facilitate the solving process of the MOP described by~(\ref{eq_MO}) and~(\ref{eq_overall_s_cst}), we introduce an additional function
 \begin{equation}\label{eq_Uc}
    \begin{split}
  U_c(  \bm{p}(k) ){ }={ }  & \; \sum_{n=1}^{N_s}
\left\{
  \begin{array}{c}
 \max\{ |s_n(k+1)- s_n(k)| -\Delta s_n,0 \}\\
  \end{array}
\right. \\
        & +  \max\{\underline{s}_n-s_n(k+1),0\}\\
   &   + \max\{s_n(k+1)-\bar{s}_n,0 \}\hspace{-0.4cm}
\left.
         \begin{array}{c}
       \\
         \end{array}
       \right\}
    \end{split}
 \end{equation}
to replace the constraints in~(\ref{eq_overall_s_cst}),
where
\begin{equation}\label{eq_p_def}
  \bm{p}(k)  { }={ }
  \left[
                  \begin{array}{ccccc}
                  \lambda(k) & p_{g_1}(k)  &  p_{g_2}(k) & \cdots &  p_{g_{N_s}}(k) \\
                  \end{array}
                \right]^T.
\end{equation}
In~(\ref{eq_Uc}), $U_c(\cdot)$ is a function of~$\bm{p}(k)$ because the function value depends on the stored energy levels that are affected by~$\bm{p}(k)$, as discussed in Remark~\ref{rmk_sn_func}.
A point $\bm{p}(k)$ satisfies the conditions in~(\ref{eq_overall_s_cst})
if and only if $U_c(  \bm{p}(k) )=0$.
With the help of~(\ref{eq_Uc}), we
define the vector-valued function
\begin{equation}\label{eq_vec_obj}
\begin{split}
\bm{F}(\bm{p}(k)){ }={ } &
\left[
   \begin{array}{cc}
    -U_{d}( \bm{p}(k)) &  -U_{g}(\bm{p}(k)) \\
   \end{array}
 \right.
 \\
    &
    \left.
      \begin{array}{cc}
    -\sum_{n=1}^{N_s}  s_n(k+1) & U_c(\bm{p}(k)) \\
      \end{array}
    \right]^T
\end{split}
\end{equation}
and introduce
Pareto optimality~\cite{ChiuCP13,MOEA}  for the ensuing discussions.

\begin{defi}[Pareto dominance]\label{def_domi}
Let~$[\cdot]_i$ denote the $i$th entry of a vector.
Consider an MOP with~$\bm{H}$ as the vector-valued objective function.
In the objective function space, a vector $\bm{u}$  dominates  another vector
$\bm{v}$ (denoted by $\bm{u} \preceq \bm{v} $)
if the condition $[\bm{u}]_i \leq [\bm{v}]_i$  holds true for all~$i$ and at least one inequality is strict.
In the decision variable space, a point
$\bm{a}$ dominates another point $\bm{b}$ with respect to~$\bm{H}$
if    $\bm{H}(\bm{a})\preceq \bm{H}(\bm{b})$.
\end{defi}

\begin{defi}[Pareto optimal solution]\label{def_P_opt}
 A point $\bm{p}$ in the decision variable space is a Pareto optimal solution if
 $\bm{p}$ is feasible and there does not exist a feasible point that dominates it.
\end{defi}

\begin{defi}[Pareto optimal set and Pareto front]\label{def_P_opt}
The set of all Pareto optimal solutions is termed the Pareto optimal set.
The image of the Pareto optimal set through the objective function~$\bm{H}$ is termed the Pareto front.
\end{defi}

The following theorem shows the equivalency of the MOP described by~(\ref{eq_MO}) and~(\ref{eq_overall_s_cst}) and an unconstrained MOP that has the vector objective function in~(\ref{eq_vec_obj}).

\begin{thm}\label{thm_Pareto_eq}
A point $\bm{p}^*(k)$ is a  Pareto optimal solution of
\begin{equation}\label{eq_MO2}
    \min_{\bm{p}(k)} \;  \bm{F}(\bm{p}(k)) \\
\end{equation}
and satisfies the condition~$U_c(\bm{p}^*(k))=0$
if and only if (``$\Leftrightarrow$'') $\bm{p}^*(k)$ is a Pareto optimal solution of~the MOP described by~(\ref{eq_MO}) and~(\ref{eq_overall_s_cst}).
\end{thm}
\begin{proof}
``$\Rightarrow$'' Since~$\bm{p}^*(k)$ satisfies the condition~$U_c(\bm{p}^*(k))=0$,  $\bm{p}^*(k)$ is a feasible point of the MOP described by~(\ref{eq_MO}) and~(\ref{eq_overall_s_cst}).
Let us denote
\begin{equation}\label{eq_G}
\bm{F}(\bm{p}(k))=
\left[
   \begin{array}{cc}
   \bm{G}^T (\bm{p}(k)) &  U_c(\bm{p}(k))   \\
   \end{array}
 \right]^T
\end{equation}
 in~(\ref{eq_vec_obj}).

We proceed by contraposition. Suppose that there exists a feasible point $\bm{p}'$ dominating $\bm{p}^*(k)$ in the MOP described by~(\ref{eq_MO}) and~(\ref{eq_overall_s_cst}), i.e., $ \bm{G}( \bm{p}') \preceq \bm{G}( \bm{p}^*(k))$.
Because $U_c(\bm{p}^*(k))=U_c(\bm{p}')=0$, we have~$\bm{F}(\bm{p}')  \preceq \bm{F}( \bm{p}^*(k)) $, which contradicts the fact that
$\bm{p}^*(k)$ is a  Pareto optimal solution of~(\ref{eq_MO2}). \\
``$\Leftarrow$''  Since $\bm{p}^*(k)$ satisfies the conditions in~(\ref{eq_overall_s_cst}), we have~$U_c(\bm{p}^*(k))=0$.
It suffices to show that $\bm{p}^*(k)$ is  Pareto optimal in~(\ref{eq_MO2}).
We use contraposition. Suppose that  there exists a point $\bm{p}'$ which satisfies the condition~$U_c(\bm{p}')=0$
and dominates  $\bm{p}^*(k)$, i.e., $ \bm{F}( \bm{p}') \preceq \bm{F}( \bm{p}^*(k))$.
Since $U_c(\bm{p}^*(k))=U_c(\bm{p}')=0$, we have~$ \bm{G} ( \bm{p}')\preceq \bm{G}( \bm{p}^*(k))$ in which $ \bm{G}$ is defined in~(\ref{eq_G}).
However, this gives contradiction because  $\bm{p}^*(k)$ is Pareto optimal in the MOP described by~(\ref{eq_MO}) and~(\ref{eq_overall_s_cst}).
\end{proof}

Based on Theorem~\ref{thm_Pareto_eq},  we can solve~(\ref{eq_MO2}) for a set of Pareto optimal solutions, and then
remove the points~$\bm{p}$ that  violate the condition~$U_c(\bm{p})=0$  to obtain Pareto optimal solutions of the MOP described by~(\ref{eq_MO}) and~(\ref{eq_overall_s_cst}).
The remaining work is to develop an algorithm to solve the MOP in~(\ref{eq_MO2}), which is addressed in the next section.

\section{Proposed Algorithm for Multi-microgrid System Design}\label{sec_design}

 Artificial immune system (AIS) algorithms
 have been proven successful in searching for Pareto optimal solutions~\cite{AIS1,AIS2}.
In this section, we adopt their basic structures to develop our algorithm used to solve~(\ref{eq_MO2}).
The terminology in an AIS is thus used: a point in the decision variable space is termed an antibody.

Here is a brief discussion on how the proposed MOIA works.
Our algorithm uses gene operations to preserve the diversity of
antibodies so that the search space can be explored.
During the iteration, dominated antibodies are removed gradually and thus nondominated antibodies are maintained in the population.
At the end of the iteration, nondominated antibodies with~$U_c(\bm{p}(k))=0$ serve as the approximation of Pareto optimal solutions.
 The solution that maximizes the minimum improvement (after normalization) in all dimensions among the approximated Pareto optimal solutions is selected as the output.
Price generation and power distribution are performed accordingly.

Fig.~\ref{fig_HMIA} presents the  pseudocode of the proposed algorithm, which
is performed at each $k$. To shorten our notation, we omit the time index~$k$ when referring to antibodies.
 The set $\mathcal{A}(t_c)$ denotes the current population, and its size is denoted by $|\mathcal{A}(t_c)|$. The $N_{nom}$ and  $N_{max}$ represent the nominal and maximum population sizes, respectively. The $t_c$ and $t_{max}$ represent the algorithm counter and the maximum iteration number, respectively.
The assignment operator ``$:=$'' is used, e.g., $A:=B$ means that we assign a new value $B$ to $A$.
Detailed discussions on the pseudocode are given as follows.

\begin{figure}
\hrulefill\\
\textbf{Input}:
\begin{itemize}
  \item MOP~(\ref{eq_MO2}).
   \item $N_{nom}$, $N_{max}$,  and $t_{max}$.
\end{itemize}
\emph{Step 1)} Initialize population $\mathcal{A}(0)$ over $[\bm{\underline{p}},\bar{\bm{p}}]$.  \\
\emph{Step 2)} Remove dominated antibodies from  $\mathcal{A}(0)$. Let $t_c=0$.

\textbf{While} $t_c\leq t_{max}$
\begin{itemize}
\item[] \emph{Step 3)} Apply gene operation to $\mathcal{A}(t_c)$ over $[\bm{\underline{p}},\bar{\bm{p}}]$.
\item[] \emph{Step 4)} Remove the antibody~$\bm{p}$ that yields the highest positive value of $U_c(\bm{p})$ from~$\mathcal{A}(t_c)$ successively until
                        the condition
                        \begin{equation*}
                          |\mathcal{A}(t_c)|\leq N_{nom}
                          \mbox{ or }
                        U_c(\bm{p})=0   \; \forall \bm{p}\in \mathcal{A}(t_c)
                        \end{equation*}
                       holds true.
\item[] \emph{Step 5)} Remove dominated antibodies from  $\mathcal{A}(t_c)$.
\item[] \emph{Step 6)} Remove the antibody that yields the least fitness from $\mathcal{A}(t_c)$ successively until the condition $|\mathcal{A}(t_c)|\leq N_{nom}$ holds true.
\end{itemize}
Let $\mathcal{A}(t_c+1):= \mathcal{A}(t_c)$ and $t_c:=t_c+1$.\\
\textbf{ End While} \\
\emph{Step 7)} Remove  antibodies~$\bm{p}$ that yield~$U_c(\bm{p})>0$  from~$\mathcal{A}(t_{max})$.\\
\textbf{Output}:
\begin{itemize}
  \item The antibody
\begin{equation*}
\bm{p}^*=\left[
           \begin{array}{ccccc}
            \lambda^* & p_{g_1}^*  &  p_{g_2}^* & \cdots &  p_{g_{N_s}}^* \\
           \end{array}
         \right]^T
\end{equation*}
that maximizes the minimum normalized improvement
in all dimensions.
\end{itemize}
\hrulefill
  \caption{Pseudocode of the proposed MOIA for the MO and DNO design.}\label{fig_HMIA}
\end{figure}

Step 1: Randomly generate the  initial population
          \begin{equation}\label{eq_ini_pop}
          \mathcal{A}(0)=\{\bm{p}_1,\bm{p}_2,...,\bm{p}_{ N_{nom}} \}
          \end{equation}
          where~$\bm{p}_j$   is a random vector  over       $[\bm{\underline{p}},\bar{\bm{p}}]$.
      The lower bound~$\bm{\underline{p}}$ and upper bound~$\bar{\bm{p}}$ can be set using~(\ref{eq_s}) and~(\ref{eq_charge_rate}), which will be illustrated in our simulations.

Steps 2 and 5:
Dominated antibodies with respect to~$\bm{F}$ are removed and nondominated antibodies are kept.
In this way, nondominated vectors~$\bm{F}(\mathcal{A}(t_c))$  can gradually approach the Pareto front as the algorithm counter~$t_c$ increases.

Step 3: Let  $N_{p}(t_c)=|\mathcal{A}(t_c)|$.
By applying the gene operation to the current population
\begin{equation}\label{eq_current_pop}
\mathcal{A}(t_c)=\{ \bm{p}_1,\bm{p}_2,\ldots,\bm{p}_{ N_{p}(t_c)} \}
\end{equation}
 we obtain
a set of newly produced antibodies denoted by
\begin{equation}\label{eq_C}
\begin{split}
\mathcal{C}{ }={ }  & \{ \bm{p}_1^1, \bm{p}_1^2  ,\ldots, \bm{p}_1^{ R(t_c)-1}    \} \cup \{ \bm{p}_2^1,\bm{p}_2^2,\ldots , \bm{p}_2^{ R(t_c)-1}\}\cup\ldots  \\
   & \cup \{\bm{p}_{ N_{p}(t_c)}^1,\bm{p}_{ N_{p}(t_c)}^2 ,\ldots,\bm{p}_{ N_{p}(t_c)}^{ R(t_c)-1} \}
\end{split}
\end{equation}
where
\begin{equation*}
 R(t_c)= \llcorner N_{max} / N_{p}(t_c) \lrcorner
\end{equation*}
represents the  clonal rate ($\llcorner \cdot \lrcorner$ represents the floor function).
After the gene operation, we
let
\begin{equation*}
  \mathcal{A}(t_c):=\mathcal{A}(t_c)\cup \mathcal{C} .
\end{equation*}

From~(\ref{eq_current_pop}) to~(\ref{eq_C}), each $\bm{p}_{i}$  is cloned  and then
 mutates to  $\bm{p}_{i}^j$. The mutant~$\bm{p}_{i}^j$ is constructed according to
\begin{equation}\label{eq_hyper-cross}
   \bm{p}_{i}^j { }={ }
    \delta \bm{p}_{i}+(1-\delta)\bm{p}_{i'}
\end{equation}
where $\delta$ represents a random number from~$[0,1]$
and  $\bm{p}_{i'}$ is
a random vector over $[\bm{\underline{p}},\bm{\bar{p}} ]$.

  Steps 4--6:  To maintain a manageable size of the population~$\mathcal{A}(t_c)$,
                we must remove low-quality antibodies. In Step~4, antibodies that have $U_c(\bm{p})>0$, i.e.,  infeasible antibodies for the MOP described by~(\ref{eq_MO}) and~(\ref{eq_overall_s_cst}),
                  are removed from the population successively. The removing procedure is performed as follows:
                 If $U_c(\bm{p}_1)>U_c(\bm{p}_2)>0$, then $\bm{p}_1$ is removed first. The procedure stops when all the antibodies~$\bm{p}$ with $U_c(\bm{p})>0$   have been removed or
                 the size of the population becomes its nominal size~$N_{nom}$.
                 After steps 4 and 5, if $|\mathcal{A}(t_c)|$ is still too large, then in~Step~6 we adopt the antibody population updating process proposed in~\cite{AIS1} to shrink the size of~$\mathcal{A}(t_c)$. In~\cite{AIS1}, antibodies are assigned with smaller fitness values when the associated objective vectors  are in a crowded region and are not  ``end'' vectors in $\bm{F}(\mathcal{A}(t_c))$.

   Step 7: Since infeasible antibodies of the MOP described by~(\ref{eq_MO}) and~(\ref{eq_overall_s_cst}) are not desired and they can be identified as those~$\bm{p}$ that have  $U_c(\bm{p})>0$, we remove them
    from the population~$\mathcal{A}(t_{max})$  in this step.

Output: The  solution $\bm{p^*}$ is selected according to
\begin{equation}\label{eq_knee}
 \bm{p}^*=\arg \; \max_{ \bm{p}  \in \mathcal{A}(t_{max})}    \;   \min_{j=1,2,3}  \frac{\bar{F}_j -  [\bm{F}(\bm{p})]_j }
 {\bar{F}_j -\underline{F}_j}
\end{equation}
where
\begin{equation}\label{eq_extreme}
    \bar{F}_j = \max_{ \bm{p}  \in \mathcal{A}(t_{max})}    \;    [\bm{F}(\bm{p})]_j   \mbox{ and }  \underline{F}_j = \min_{ \bm{p}  \in \mathcal{A}(t_{max})}    \;    [\bm{F}(\bm{p})]_j.
\end{equation}
The ratio related to index $j$ in~(\ref{eq_knee}) can represent the normalized improvement in the $j$th dimension~\cite{IET_CTA_14}.
The antibody $\bm{p^*}$ selected by~(\ref{eq_knee}) can be characterized as follows:
 it maximizes the minimum improvement in order not to advantage one particular objective.

\section{Numerical Results}\label{sec_sim}

In this section, we describe numerical simulations that have been carried out to illustrate our proposed multiobjective methodology.
Suppose that we have $N=3$ and $N_s=2$, i.e., 3 microgrids in a network and among them, 2 microgrids possess energy storage systems.
Let $\bar{s}_1=250$ and $\bar{s}_2=200$ (kWh) be the maximum storage capacity in~(\ref{eq_storage_capacity}).
Let $\Delta s_n = 10\%  \bar{s}_n $ be the limits on the rate in~(\ref{eq_charge_rate}) for $n=1,2$,
and~$\underline{s}_1=125$ and $\underline{s}_2=100$ (kWh) be the secure energy levels in~(\ref{eq_Obj_iso}).

For the microgrids, the power demand function~$f_{d_n}(\cdot)$ in~(\ref{eq_p_d}) has been set as
\begin{equation}\label{eq_pd_func}
  f_{d_n}(\lambda(k),b_n(k))= ( 1+h_n(\lambda(k))) b_n(k)
\end{equation}
where
\begin{equation}\label{eq_h_func}
\begin{split}
 h_1(\lambda(k)){ } ={ } & 0.01 \lambda(k)^2-0.12\lambda(k) +0.26   \\
 h_2(\lambda(k)){ } ={ } & -0.01 \lambda(k)^2+0.13    \\
h_3(\lambda(k)){ } ={ } &    -0.01  \lambda(k)^2+0.02\lambda(k)+0.08    .
\end{split}
\end{equation}
In~(\ref{eq_h_func}),
the values of the functions $h_n(\cdot)$  are positive/negative (in percentage)
at a low/high price~$\lambda(k)$.
See Fig.~\ref{fig_h1} for a graphic illustration of~$h_1$.
Therefore,
the power demand in~(\ref{eq_pd_func}) increases/decreases by  $|h_n(\lambda(k))| b_n(k)$ when a low/high price signal occurs.
The curve in Fig.~\ref{fig_h1} is termed a demand curve~\cite{demand_curve1,demand_curve2} and in practice, it can be constructed using collected data
and the regression analysis.
The base load~$b_n(k)$ from~\cite{EMS} has been used, shown in Fig.~\ref{fig_base_load}.

\begin{figure}
  \centering
  \includegraphics[width=7cm]{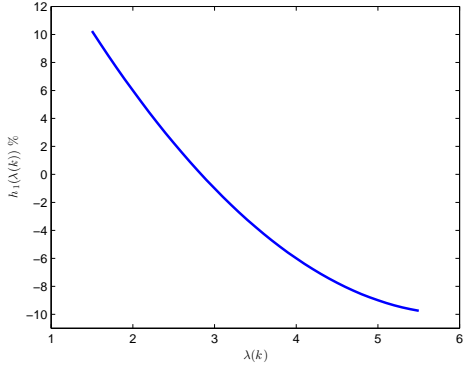}\\
  \caption{Graph of  $h_1$ in~(\ref{eq_h_func}).}\label{fig_h1}
\end{figure}

\begin{figure}
\centering
\includegraphics[width=8cm]{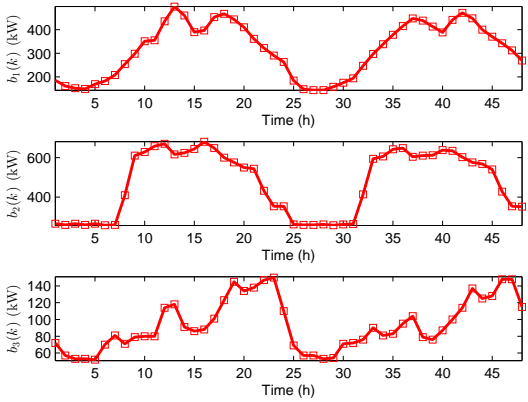}
\caption{Base load of microgrids $b_n(k)$ in our simulations.}\label{fig_base_load}
\end{figure}

 \begin{figure*}
\begin{equation*}
\begin{array}{ccc}
\hspace{-0.5cm}\includegraphics[width=6cm]{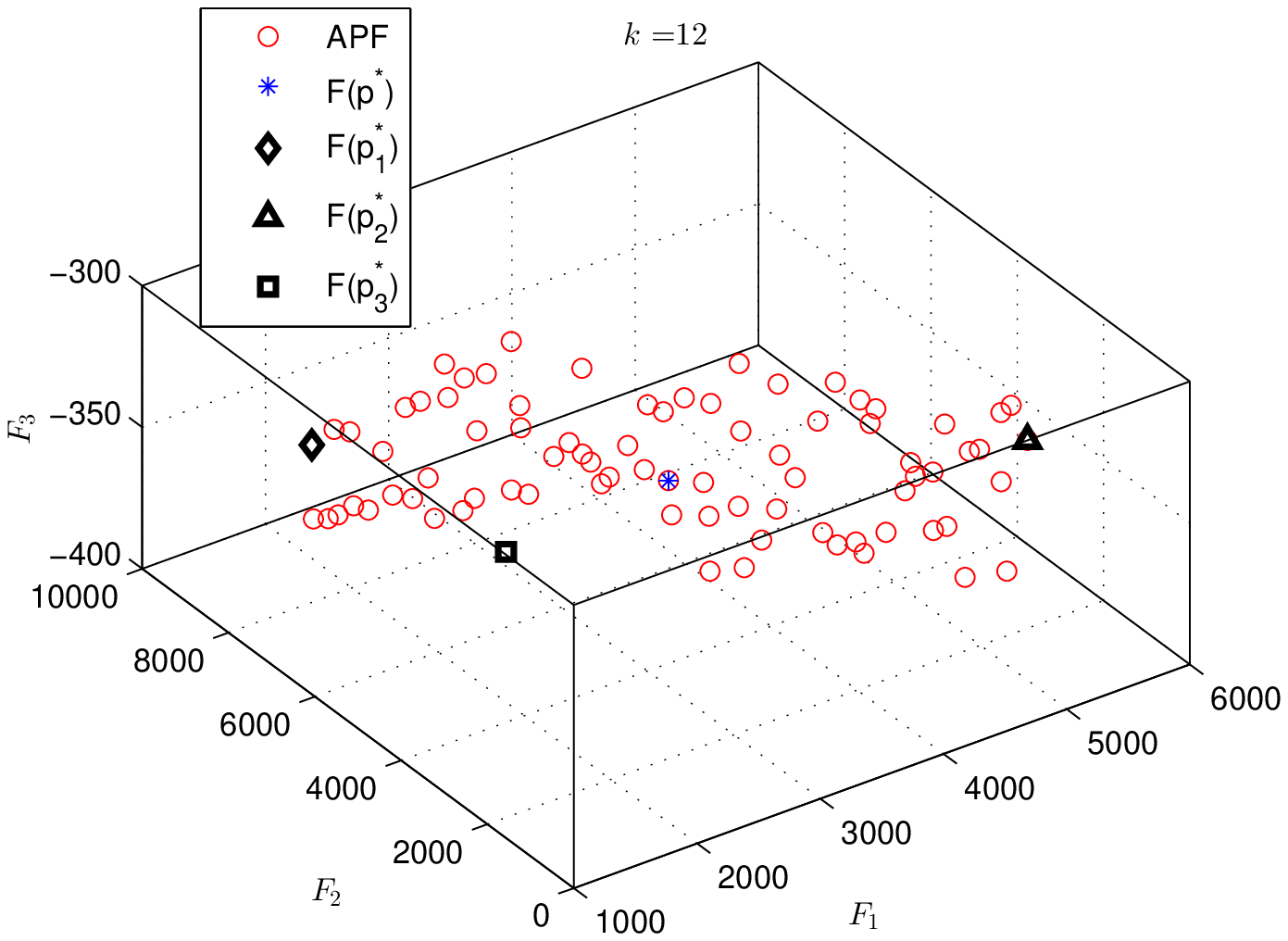} & \includegraphics[width=6cm]{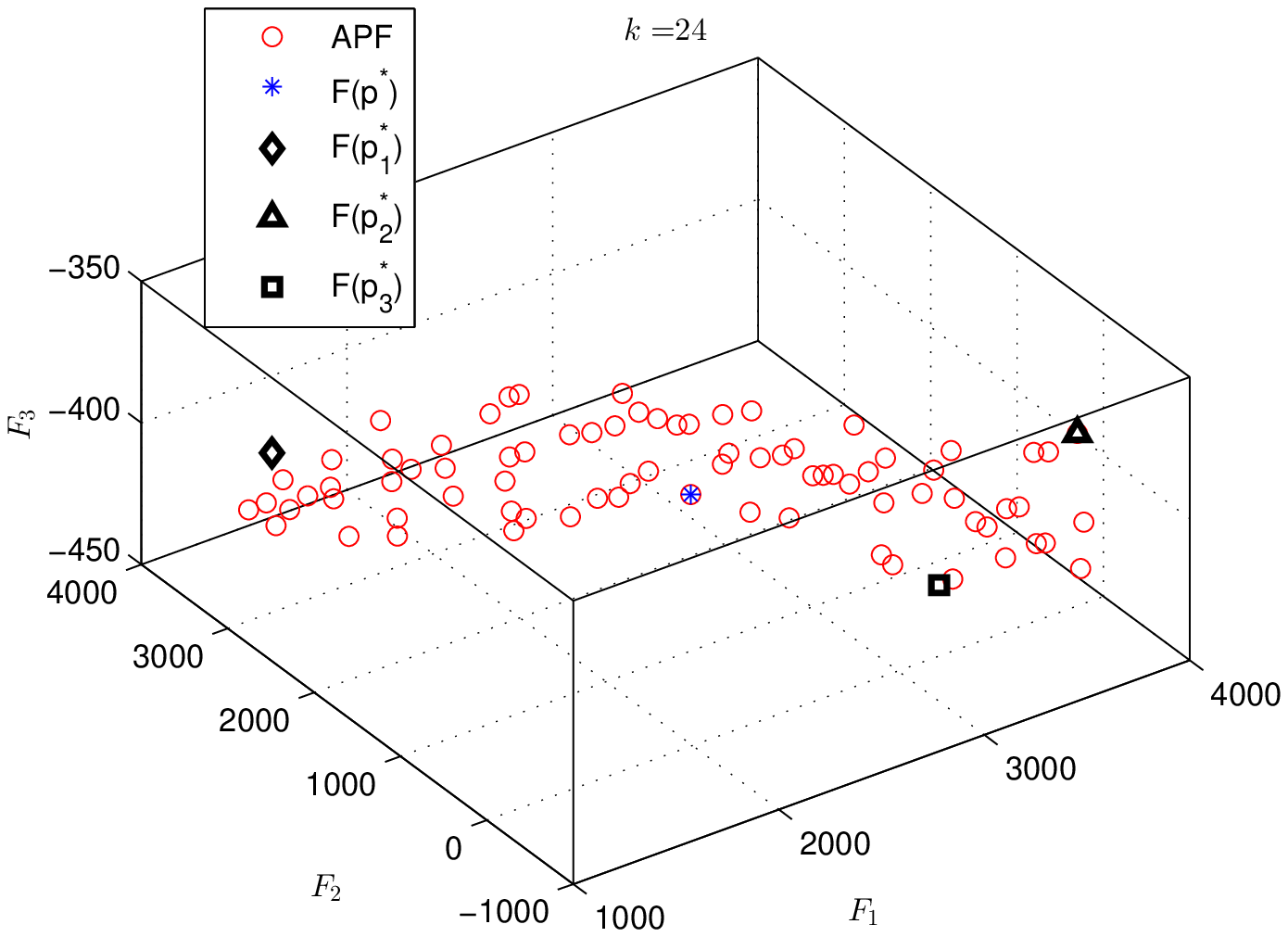} & \includegraphics[width=6cm]{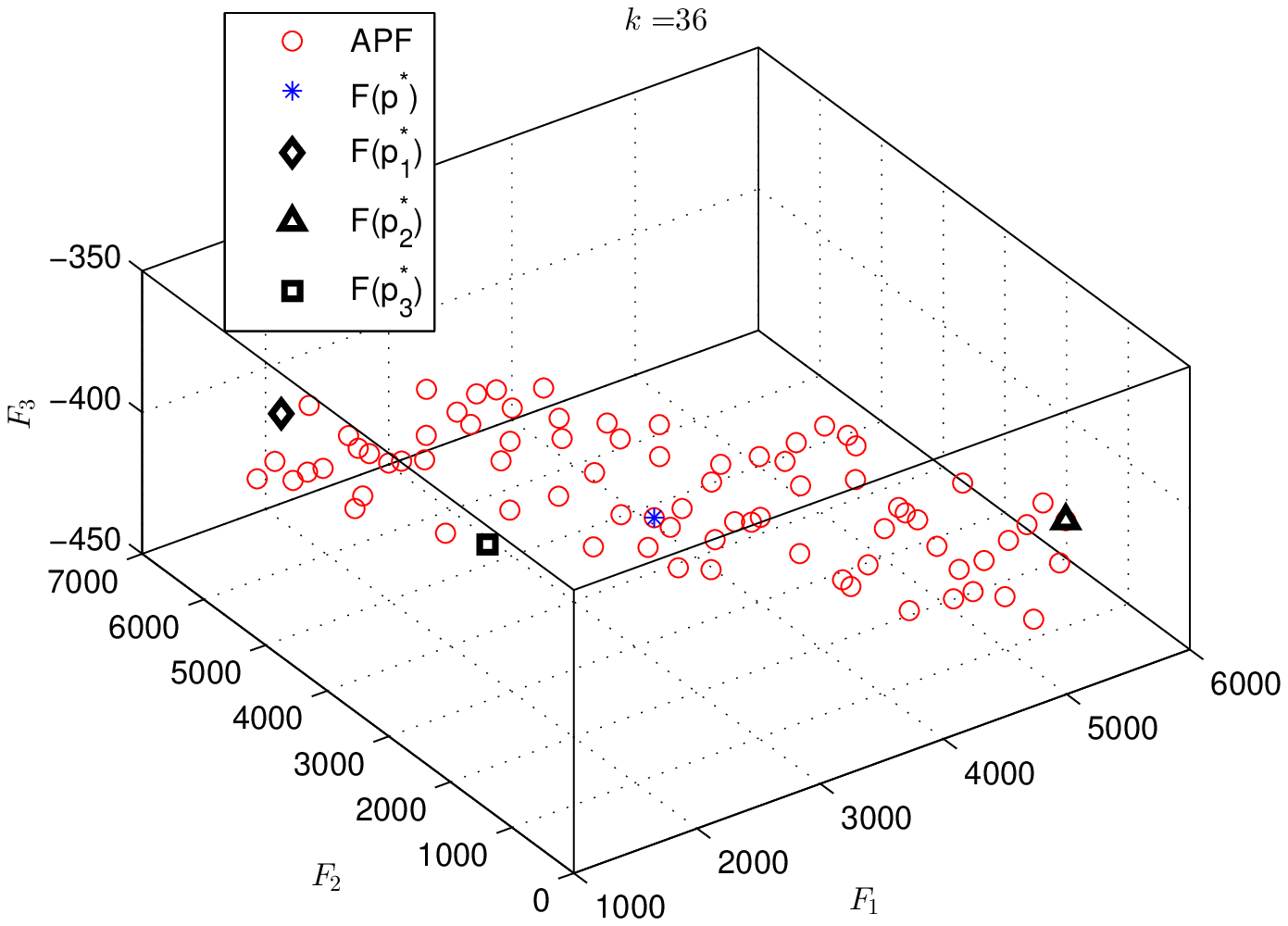} \\
  \mbox{(a)} & \mbox{(b)} & \mbox{(c)}
\end{array}
\end{equation*}
\caption{APFs and nondominated vectors $\bm{F}(\bm{p}^*)$  sampled at (a) time $k=12$; (b) time $k=24$; and (c) time $k=36$.
The~$\bm{p}_j^*$ is the point that achieves the minimum~$ \underline{F}_j$ defined in~(\ref{eq_extreme}).
}\label{fig_APF}
\end{figure*}

For the overall utility of microgrids, we let~\cite{price_em,faranda2007load,samadi2010optimal}
\begin{equation}\label{eq_Obj_d_function}
\begin{split}
   &  U_{d}(p_{d_1}(k), \ldots ,p_{d_N}(k),\lambda(k))\\
= { }    & \sum_{n=1}^N
( U(p_{d_n}(k),\omega_n)- \lambda(k) p_{d_n}(k)    )
\end{split}
\end{equation}
where
\begin{equation}\label{eq_U}
\begin{split}
  &  U(p_{d_n}(k),\omega_n) \\
 ={ }  & \left\{
  \begin{array}{ll}
    \omega_n p_{d_n}(k) -\frac{\alpha}{2}p_{d_n}(k)^2     , & \hbox{if } 0\leq p_{d_n}(k) \leq  \frac{\omega_n}{\alpha},   \\
   \frac{\omega_n}{\alpha}, & \hbox{if }  p_{d_n}(k) \geq \frac{\omega_n}{\alpha}.
  \end{array}
\right.
\end{split}
\end{equation}
The terms
$ U(p_{d_n}(k),\omega_n)$ and $\lambda(k) p_{d_n}(k)$ in the summation
are interpreted as the value and the cost derived from consuming power~$p_{d_n}(k)$ at microgrid $n$,  respectively.
For the power grid, the utility function in~(\ref{eq_Obj_g}) has been set as~\cite{price_em,mohsenian2010optimal,samadi2010optimal}
\begin{equation}\label{eq_Obj_g_function}
     U_{g}(p_{g}(k),\lambda(k))= \lambda(k)p_{g}(k)  -(a(k)p_{g}(k)^2+b(k)p_{g}(k)+c(k)  )
\end{equation}
 where the first and second terms on the right-hand side represent the value derived from power generation and
the generation cost, respectively. The values of $\omega_n,\alpha,a(k),b(k),$ and $c(k)$ in~(\ref{eq_U}) and~(\ref{eq_Obj_g_function}) can be found in~\cite{samadi2010optimal}
and in practice, they can be obtained by statistical analysis~\cite{yu2012statistical}.
For the ISO, the secure energy level  $\underline{s}_n=\bar{s}_n/2$ in~(\ref{eq_Obj_iso}) has been chosen.

The values
$N_{nom}=80$, $N_{max}=320$,  and $t_{max}=200$ have been used as the inputs to the proposed algorithm presented in Fig.~\ref{fig_HMIA}.
The bounds~$\bm{\underline{p}}$ and~$\bm{\bar{p}}$ have been assigned as follows.
Because there are strong regulations on market pricing in practice, we suppose that
\begin{equation}\label{eq_price_bd}
    \lambda(k)\in [\underline{\lambda},\bar{\lambda}]
\end{equation}
with  $\underline{\lambda}=1.5$ and $\bar{\lambda}=5.5$
to mimic such regulations  in the simulations.
Note that for $n\in \mathcal{N}_s$,~(\ref{eq_s}) and~(\ref{eq_charge_rate}) imply that
         \begin{equation*}
         \begin{split}
           p_{g_n}(k) { } \leq{ }&     \Delta s_n +  p_{d_n}(k) -v_n(k)  \mbox{ and }\\
           p_{g_n}(k) { } \geq { }   & -\Delta s_n +  p_{d_n}(k) -v_n(k).
         \end{split}
         \end{equation*}
Based on~(\ref{eq_h_func}) and~(\ref{eq_price_bd}),
 we have
\begin{equation*}
h_n(\bar{\lambda}) \leq  h_n(\lambda(k))  \leq   h_n(\underline{\lambda})
\end{equation*}
 and, therefore,
                \begin{equation}\label{eq_bd_p}
         \begin{split}
           p_{g_n}(k) { } \leq{ }&     \Delta s_n +  (1+ h_n(\underline{\lambda}))b_n(k) -v_n(k)  \mbox{ and }\\
           p_{g_n}(k) { } \geq { }   & -\Delta s_n +  (1+h_n(\bar{\lambda}))b_n(k) -v_n(k).
         \end{split}
         \end{equation}
Referring to~(\ref{eq_p_def}),~(\ref{eq_price_bd}), and~(\ref{eq_bd_p}), we let
  \begin{equation*}
{\scriptsize
\begin{split}
  [\bm{\underline{p}}]_1 { }={ }  &\underline{\lambda}, [\bm{\underline{p}}]_j= -\Delta s_{j-1} +  (1+h_{j-1}(\bar{\lambda}))b_{j-1}(k)-v_{j-1}(k) , j=2,3, \\
[\bm{\bar{p}}]_1{ }={ }&  \bar{\lambda}, [\bm{\bar{p}}]_j= \Delta s_{j-1} +  (1+  h_{j-1}(\underline{\lambda})    )b_{j-1}(k) -v_{j-1}(k),  j=2,3.
\end{split}
}
\end{equation*}
be the  bounds for antibody generation at time~$k$.

\begin{figure*}
\begin{equation*}
\begin{array}{ccc}
\hspace{-0.5cm}\includegraphics[width=5.5cm]{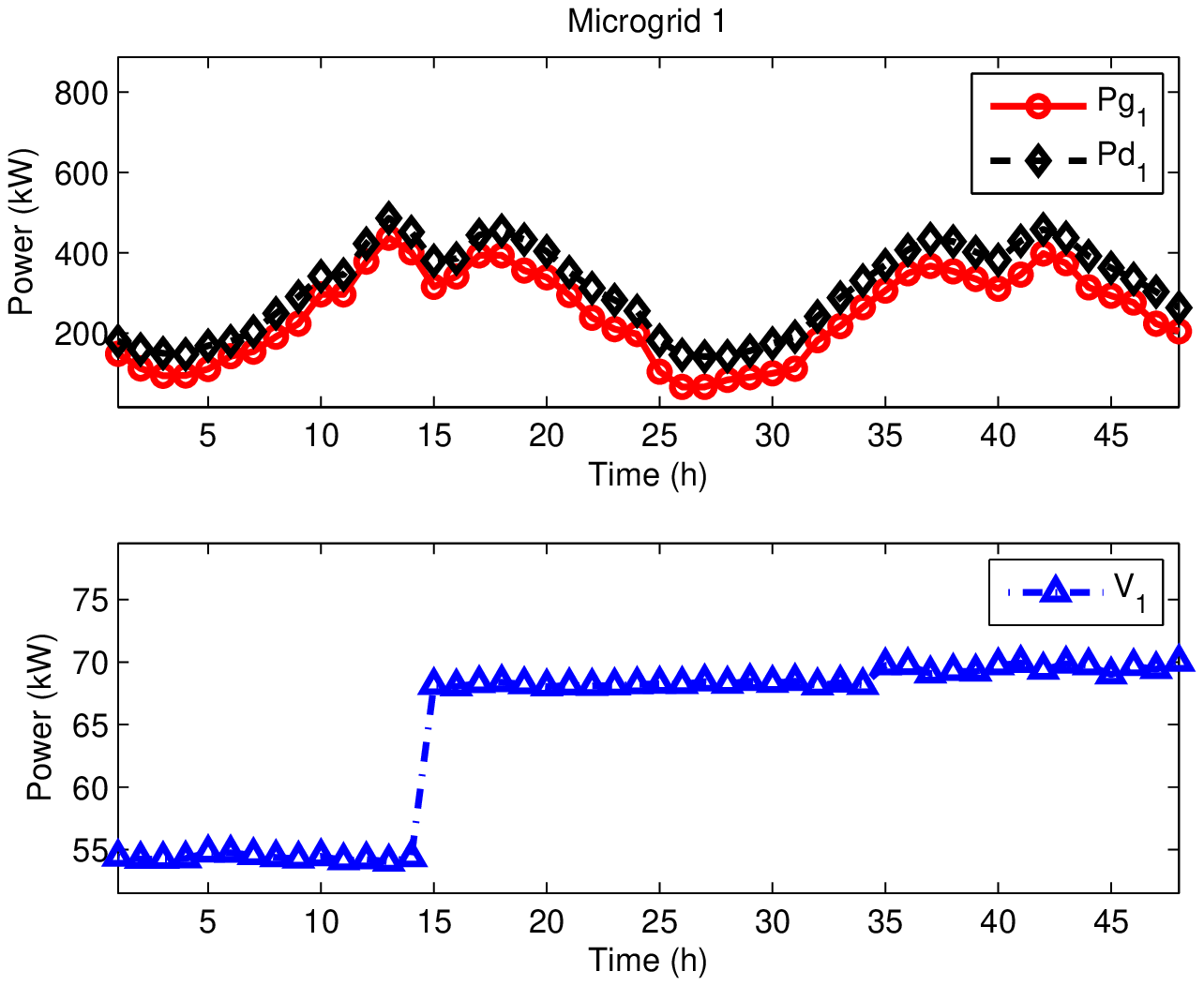} & \includegraphics[width=5.5cm]{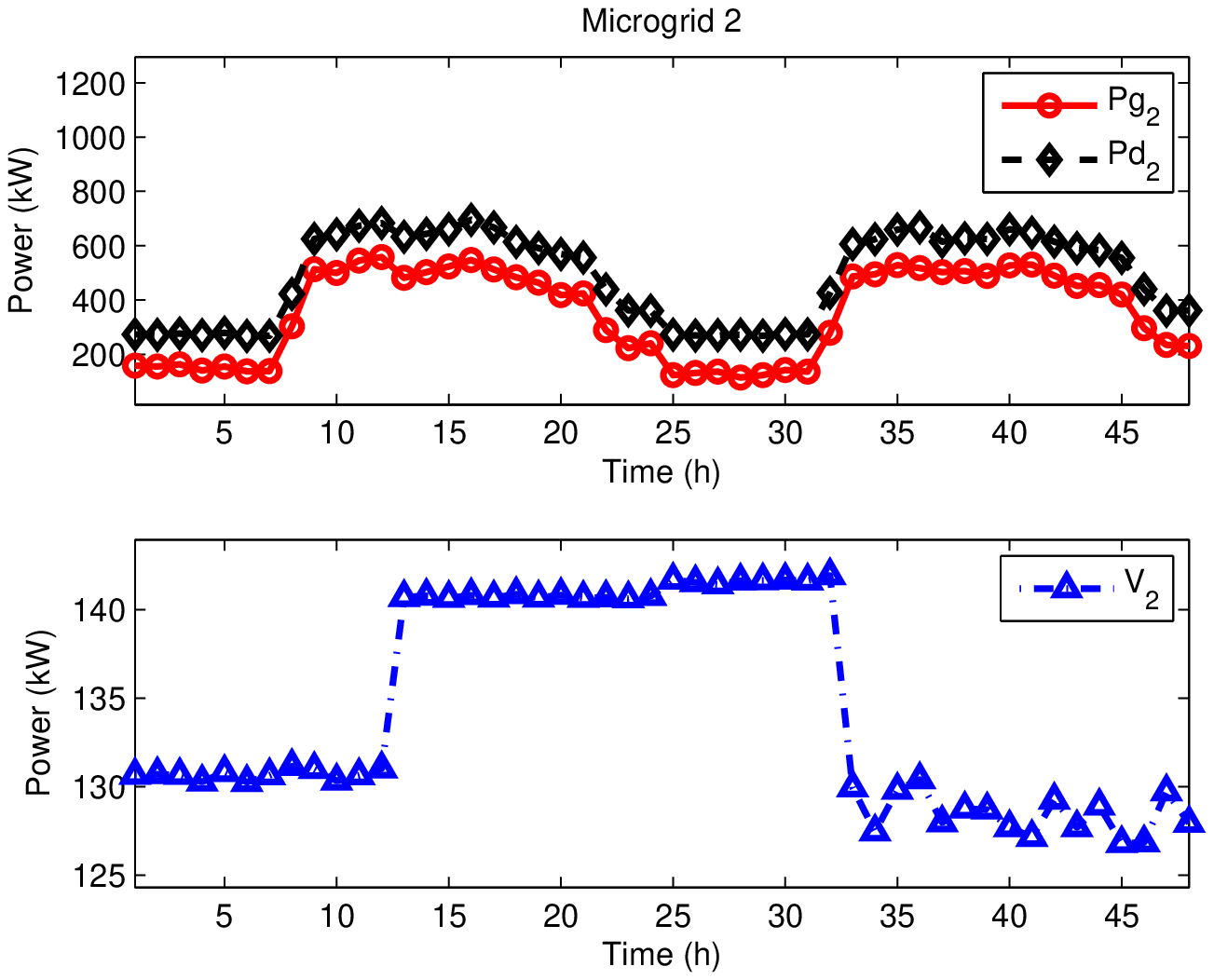} & \includegraphics[width=5.5cm]{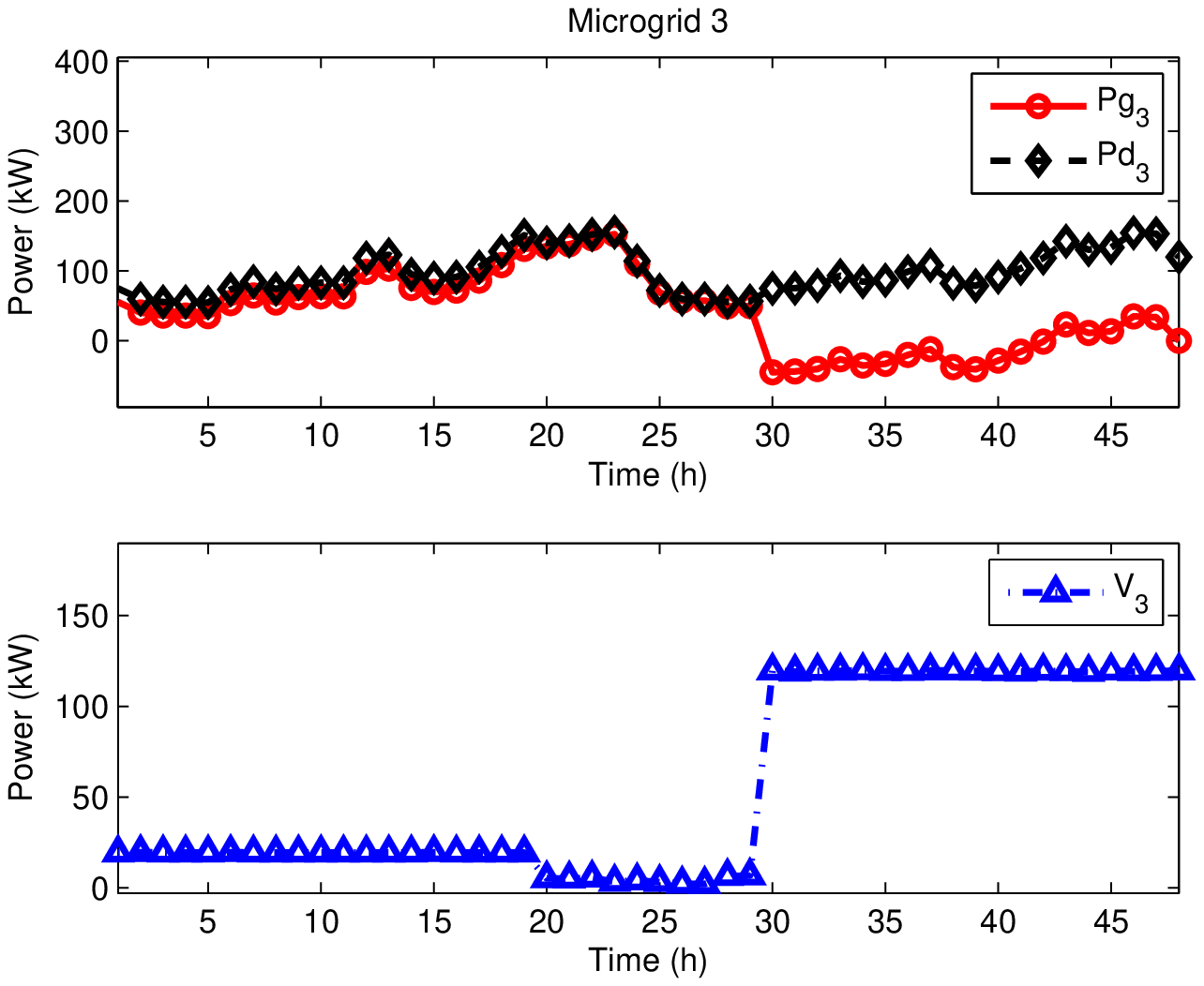} \\
  \mbox{(a)} & \mbox{(b)} & \mbox{(c)}
\end{array}
\end{equation*}
\caption{Relation between the distributed power $p_{g_n}(k)$, power demand $p_{d_n}(k)$,  and power input $v_n(k)$ provided by RESs in (a) microgrid 1, (b) microgrid 2, and (c) microgrid 3.
}\label{fig_g_d_v}
\end{figure*}

Fig.~\ref{fig_APF} presents samples of APFs, illustrating how one objective affects the others.
In our scenario, if a vector on one APF has an entry that achieves an extreme value, then the associated design is regarded as being in favor of one particular participant.
To yield a fair design, a nondominated vector after normalization  should be positioned away from extremes equally in all dimensions.
This has been achieved using~(\ref{eq_knee}) to obtain the vectors $\bm{F}(\bm{p}^*)$ in Figs.~\ref{fig_APF} (a)--(c).
As a result, all the vectors  lie in the ``middle'' of the APFs graphically
and, therefore, the proposed multiobjective approach can provide a reasonable way to produce a fair design to all participants.

 Fig.~\ref{fig_g_d_v} provides an overall view on the relation between  $p_{g_n}(k)$, $p_{d_n}(k)$  and $v_n(k)$.
The demand $p_{d_n}(k)$ responses to the changes of $p_{g_n}(k)$ and $v_n(k)$, and
in most cases, $p_{d_n}(k)$ is larger than the supply $p_{g_n}(k)$ because of the existence of $v_n(k)$.
The distributed power $p_{g_n}(k)$ has been adjusted according to the power input $v_n(k)$:
$p_{g_n}(k)$ increases upon decreasing $v_n(k)$, and  decreases when $v_n(k)$ increases.

Fig.~\ref{fig_energy_state} shows the energy management at each microgrid.
As shown in the proof of Theorem~\ref{thm_stable}, we can keep the stored energy~$s_n(k)$ at a constant level if desired.
 This implies that when an SOP regarding the interest of the ISO is considered, the optimal objective value $\sum_{n =1}^{N_s}  \bar{s}_n$ can always be achieved after a period of time.
In our scenario, however, since multiple objectives were considered,
the stored energy levels~$s_n(k)$ vibrated  in response to the time-varying $p_{g_n}(k)$, $p_{d_n}(k)$, $v_n(k)$, and $\lambda(k)$.

Finally, as discussed in~\cite{chiu_sg}, the vibration of price
plays an important role in  energy management at microgrids.  It can be observed from Fig.~\ref{fig_cpr_price} that
the proposed approach does yield prominent price vibration.
Ideally, when a large portion of the demand is shiftable and an aggregate utility function is used, high prices yield peak load shaving and storage discharging
while low prices yield valley load filling and storage charging~\cite{SGComm_Nov_12}.
These phenomena are not prominent in our simulations mainly because
only a small portion of the demand was assumed to be shiftable
and a multiobjective approach was used.
In a multiobjective scenario, for example, high prices may not yield storage discharging
because the ISO does not consider price values and simply desires  high, secure energy levels.

The multi-microgrid system operated in consideration of  the interests of all participants.
Price generation and power distribution were performed according to the selected Pareto solutions.
As a result, the physical constraints were satisfied during the operation, and the power demand and stored energy levels were adjusted properly.
The validity of the proposed multiobjective approach was thus confirmed.

\begin{figure}
\centering
\includegraphics[width=8cm]{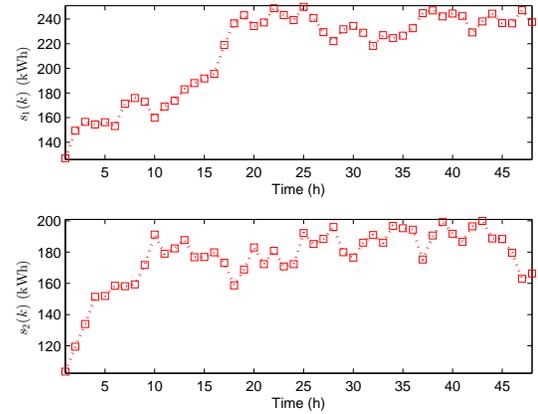}
\caption{ The stored energy levels $s_n(k)$ at the microgrids.
  }\label{fig_energy_state}
\end{figure}

\begin{figure}
\centering
\includegraphics[width=6cm]{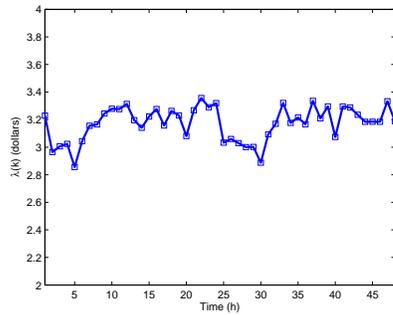}
\caption{Price signal $\lambda(k)$ generated from the MO using the proposed multiobjective approach.
}\label{fig_cpr_price}
\end{figure}

\section{Conclusion}\label{sec_con}

A multi-microgrid system design in consideration of interests of the microgrids, power grid, and ISO has been considered in this paper.
We believe that a fair scheme can promote active participation, which provides foundation for
new and interactive service in the future grid.
To this end, we  have formulated the design problem as an MOP
and proposed the MOIA to solve it, leading to a
 multiobjective design approach.
This approach maximizes the utilities of the microgrids, power grid, and ISO simultaneously.
Pareto optimal
market prices and power distribution can then be produced for the MO and DNO, respectively.
Our multiobjective approach is  general and flexible.
We have argued that  the proposed methodology can be readily applied to other scenarios with objectives and constraints different from those considered in this paper.
This is because the proposed MOIA
that searches for Pareto optimal designs
is not developed based on particular structures of the objective and constraint functions.


\end{document}